\documentclass[10pt, a4paper, reqno]{amsart}
\usepackage{mathrsfs}
\usepackage{color}
\usepackage{bbm}
\usepackage{amsmath,amscd,amssymb,latexsym}
\usepackage[all]{xy}

\input xypic

\evensidemargin=\oddsidemargin
\DeclareMathVersion{can}
\DeclareMathAlphabet{\can}{OT1}{cmss}{m}{n}
\newtheorem{thm}{Theorem}[section]
\newtheorem{cor}[thm]{Corollary}
\newtheorem{lem}[thm]{Lemma}

\newtheorem{exa}[thm]{Example}
\theoremstyle{definition}

\numberwithin{equation}{section}

\newcommand{\ord}{\operatorname{ord}}

\newcommand{\Tr}{\operatorname{Tr}}

\begin{document}
\title[]{ A class of  functions with low-valued walsh spectrum}

\author[F. Li] {Fengwei Li}
\address{\rm
 School of Mathematics and Statistics, Zaozhuang University, Zaozhuang,
277160, P. R. China; State Key Laboratory of Cryptology, P.O. Box 5159, Beijing, 100878, P. R. China}
 \email{lfwzzu@126.com}

\author[Y. Wu]{Yansheng Wu}
\address{\rm Department of Mathematics, Ewha Womans University, 52, Ewhayeodae-gil, Seodaemun-gu,
Seoul, 03760, South Korea} 
\email{wysasd@163.com}

\author[Q. Yue] {Qin Yue}
\address{\rm  
Department of Mathematics, Nanjing University of Aeronautics and Astronautics,
Nanjing, 211100, P. R. China}
\email{yueqin@nuaa.edu.cn}

\thanks{The paper was supported by National Natural Science Foundation of China (Nos.
11601475, 61772015) and the Foundation of Science and Technology on Information Assurance Laboratory (No. KJ-17-010).}

\begin{abstract}

Let  $l\equiv 3\pmod 4$, $l\ne 3$,  be a prime, $N=l^2$, $f=\frac{l(l-1)}2$ the multiplicative order of a prime $p$ modulo $N$, and $q=p^f$. In this paper, we  investigate the Walsh spectrum of the monomial functions $f(x)=\Tr_{q/p}(x^{\frac{q-1}{l^2}})$ in index two case.
In special, we explicitly present the  value distribution  of the Walsh transform of $f(x)$ if $1+l=4p^h$, where $h$ is a class number of $\Bbb Q(\sqrt{-l})$.

 \end{abstract}
\maketitle


\section{Introduction}

 Walsh transform over finite fields is a basic tool in research of properties of cryptographic functions.  
 The important information about  cryptography  can be obtained from the study of the Walsh transform. A long-standing problem about the Walsh transform is to find functions with a few Walsh transform values and determine its distribution. There are a few monomial
functions with three-valued Walsh transform for special exponents \cite{CCD, H, NH1, NH2, WYL}.
The results on functions with at least four-valued Walsh transform were obtained in \cite{FL, HHKZ, HR, LTW, NH3}. The research  progress on a  few Walsh transform values can be referred to literatures \cite{ HK, LF1,  LHKT, ZLH} and
the references therein.
 Walsh transform is also closely related to fourier transform, Gauss period,  and     the weight distribution of a cyclic code \cite{ HY, LY, LYL, P,   YCD, ZD, ZZD} et.al..
Next let us introduce the definitions of the Walsh transform of a function over a finite field.

Let $\Bbb F_q$ be a finite field with $q$ elements, where $q$ is a  power of prime $p$.  The trace function from $\Bbb F_{q}$ onto $\Bbb F_p$ is defined by $\Tr_{q/p}(x)=x+x^p+\cdots+x^{q/p}$, $ x\in \Bbb F_{q}$.
Let $f(x)$ be a function from $\Bbb F_{q}$ to $\Bbb F_p$. The Walsh transform of $f(x)$ is defined by
$$\widehat{f }(b)= \sum_{x\in \Bbb F_{q} } \zeta_p^{ f(x)+\Tr_{q/p}(bx)},b\in \Bbb F_{q},$$ where $\zeta_p$
 is a complex primitive $p$-th root of unity. The multiset $\{\widehat{f }(b): b\in \Bbb F_{q} \}$ is called the Walsh spectrum of $f(x)$ over $\Bbb F_q$. If  $f (x)$ is a monomial function of the form $\Tr_{q/p}(ax^d)$ with $a \in \Bbb F_{q}$ and $d$ is a positive integer, then the Walsh transform of $f(x)$ gives the cross correlation values of an $m$-sequence and their $d$-decimations \cite{CCD, CD, HK, L, ZLFG}.

 This paper is organized as follows. In Section 2,
 we give several results about Gaussian sums in index two case. In Section 3, let $l\equiv 3\pmod 4$, $l\ne 3$, be a prime, $N=l^2$,   $f=\frac{l(l-1)}2$  the multiplicative order of a prime $p$ modulo $N$, and $q=p^f$. We  give the Walsh spectrum of an index two function $f(x)=\Tr_{q/p}(x^{\frac{q-1}{l^2}})$  over
$\Bbb F_q$. In special, we explicitly present the  value distribution  of the Walsh transform of $f(x)$ when $1+l=4p^h$, where $h$ is the class number of $\Bbb Q(\sqrt{-l})$. In  Section 4, we conclude this paper.

\section{ Preliminaries }

Let  $\Tr_{q/p}$ be the trace function from $\Bbb F_{q}$ to $\Bbb F_{p}$.
 An additive character of $\Bbb F_{q}$  is a nonzero function $\psi$ from $\Bbb F_{q}$
to the set of complex numbers such that $\psi(x+y) = \psi(x)\psi(y)$ for any pair $(x, y) \in \Bbb F_{q}^2$. For each $b \in \Bbb F_{q}$, the function
\begin{equation*}
\psi_b(c)=\zeta_{p}^{\Tr_{q/p}(bc)} \mbox { for all } b \in \Bbb F_{q} \end{equation*}
defines an additive character of $\Bbb F_{q}$, where $\zeta_{p}=e^{\frac{2\pi\sqrt{-1}}p}$ denotes the $p$-th primitive root of unity.
 When $b=1$, the character $\psi_1$ is called canonical additive character of $\Bbb F_{q}$. It is well know that
 \begin{equation*}
\sum_{c\in \Bbb F_q}\psi_b(c)=0 \mbox{ for } b\neq 0.
\end{equation*}
A multiplicative character of $\Bbb F_{q}$  is a nonzero function $\chi$ from $\Bbb F_{q}$  to the set of complex numbers such that $\chi(xy) =
\chi(x)\chi(y)$ for all pairs $(x, y)\in\Bbb F_{q}^{\ast2}$.  Let $\alpha$ be a fixed primitive element of $\Bbb F_{q}$. For each $j = 1, 2, \ldots, q-1$,
the function $\chi_j$ with
 \begin{equation}\chi_j(\alpha^k) = \zeta_{q-1}^{jk} \mbox { for } k = 0, 1, \ldots, q-2 \end{equation}
defines a multiplicative character with order $\frac {q-1}{gcd (q-1,j)}$ of $\Bbb F_{q}$, where $\zeta_{q-1} $ denotes the $(q-1)$-th primitive root of unity.

Let $q$ be odd and $j = \frac{q-1}2$ in (2.1), we then get a multiplicative character denoted by $\eta$ such that $\eta(c) = 1$ if $c$ is the square of an
element and $\eta(c) = -1$ otherwise. This $\eta$ is called the quadratic character of $\Bbb F_{q}$.

Let $\chi$ be a multiplicative character with order $k$ with $k|(q-1)$ and $\psi$  an additive character of $\Bbb F_{q}$. Then the Gaussian
sum $G(\chi, \psi)$ of order $k$ is defined by
$$G(\chi,\psi)=\sum_{x\in \Bbb F_{q}^{*}}\chi(x)\psi(x).$$
Since $G(\chi, \psi_b) = \bar{\chi} (b)G(\chi, \psi_1)$, we just consider $G(\chi, \psi_1)$, briefly denoted as $G(\chi)$, in the sequel.
\begin{lem}\cite{LN}
Let $\psi$ be a nontrivial additive character of $\Bbb F_{q}$ and $\chi$ a multiplicative character of $\Bbb F_{q}$ of order $s=\gcd(n,q-1)$. Then
$$\sum_{x\in \Bbb F_{q}}\psi(ax^{n}+b)=\psi(b)\sum_{j=1}^{s-1}\bar{\chi}^{j}(a)G(\chi^{j},\psi)$$
for any $a,b\in \Bbb F_{q}$ with $a\neq 0$.
\end{lem}
In general, the explicit determination of Gaussian sums is also a difficult problem.
For future use, we state the quadratic  Gaussian sums here.

\begin{lem}\cite{LN}
Suppose that $q=p^{f}$ and $\eta$ is the quadratic multiplicative character of $\Bbb F_{q}$, where $p$ is an odd prime. Then  $$G(\eta)=\left\{
\begin{array}{ll}
(-1)^{f-1}\sqrt{q}, & \mbox{if}\ p\equiv 1\pmod{4},\\
(-1)^{f-1}(\sqrt{-1})^{f}\sqrt{q}, & \mbox{if}\ p\equiv 3\pmod{4}.\\
\end{array}
\right.$$

\end{lem}

Let $\Bbb Z/N\Bbb Z=\{0,1,\ldots, N-1\}$ be the ring of integers modulo $N$ and  $(\Bbb Z/N\Bbb Z)^*$ a multiplicative group consisting of  all invertible elements in $\Bbb Z/N\Bbb Z$. If $\langle p\rangle $ is a cyclic subgroup with a generator $p$ of the group $(\mathbb{Z}/N\Bbb Z)^\ast$ such that  $[(\mathbb{Z}/N\mathbb{Z})^\ast: \langle p\rangle] = 2$
 and $-1\notin \langle p\rangle\subset
(\mathbb{Z}/N\mathbb{Z})^\ast$, which is the
so-called  ``quadratic residues" or ``index 2"case, Gaussian sums are  explicitly determined, see \cite{YX} and its references for
details. We list some results \cite{DY,YX} in the index 2 case below.

\begin{lem}\cite{ YX} Let  $l\equiv 3\pmod 4$ be a prime, $l\ne 3$, $m$ a positive integer,  $N=l^m$, $f=\frac{l(l-1)}2$ the multiplicative order of a prime $p$ modulo $N$, and $q=p^f$.
Suppose that
 $\chi$ is  a primitive multiplicative character of order $N$ over  $\Bbb F_q$.

(1)  For $ 1\le i\le N-1$, let $i=ul^t$,
 $0\le t\le  m-1$  and $\gcd(u,l)=1$. Then
$$G(\chi^i)=\left\{\begin{array}{ll}
G(\chi^{l^t}) & \mbox{ if } u\in \langle p \rangle \subset \mathbb{Z}_N^\ast,\\
\overline {G(\chi^{l^t})} & \mbox{ if } u\notin \langle p \rangle \subset \mathbb{Z}_N^\ast.\end{array}\right.$$

(2) For $0\le t\le m-1$,
$$G(\chi^{l^t})=p^{\frac{f-hl^t}2}\left(\frac{a+b\sqrt{-l}}2\right)^{l^t},$$
where $h$ is the ideal class number of $\mathbb{Q}(\sqrt{-l})$, $a,b$ are integers given by
\begin{eqnarray}\left\{\begin{array}{l}
a^2+lb^2=4p^h  \\
a\equiv -2p^{\frac{l-1+2h}4} \pmod l. \end{array}\right.\end{eqnarray}

\end{lem}

Let $\mathcal O=\Bbb Z[\sqrt{-l}]$ be all algebraic integers in $\Bbb Q(\sqrt{-l})$.  Then $p\mathcal O={\mathcal P}_1 {\mathcal P}_2$, where ${\mathcal P}_1=\langle\frac{a+b\sqrt {-l}}2, p\rangle$ and ${\mathcal P}_2=\langle\frac{a-b\sqrt {-l}}2, p\rangle$. In fact, the multiplicative  character $\chi$ is correspondent to ${\mathcal P}_2$ (see
\cite{La}).

\section{Explicit valuations of Walsh spectrum of  $f(x)=\Tr_{q/p}(x^{\frac{q-1}{l^2}})$}

Let $l\equiv 3\pmod 4$, $l\ne 3$, be a prime, $N=l^2$,  and $f=\frac{l(l-1)}2$  the multiplicative order of a prime $p$ modulo $N$,  i.e., $f$  the smallest positive integer such that $p^f\equiv 1 \pmod N$.   Let $\Bbb F_q$ be a finite field with $q=p^f$ elements. In this section, we shall  give the value distributions  of  the Walsh transform of $f(x)=\Tr_{q/p}(x^{\frac{q-1}N})$.
 We always  denote $\alpha$  a primitive element in $\Bbb F_q^*$ and $\beta=\alpha^{\frac{q-1}{N}}$ an $N$-th primitive root of unity in  $\Bbb F_{q}$. Let $K=\langle \alpha^N\rangle$ denote the subgroup of  $\Bbb F_q$ generated by $\alpha^N$. Then $\Bbb F_q^*=\left\langle \alpha \right\rangle=\bigcup_{\mbox{\tiny$
\begin{array}{c}
{i=0}
\end{array}$
}}^{\mbox{\tiny$
\begin{array}{c}
N-1
\end{array}$
}} \alpha^iK.$

Define
$C_i^{(N,q)}=\alpha^i\langle \alpha^N\rangle$ for $i = 0, 1,\ldots , N -1$. The cyclotomic
numbers of order $N$ are defined by
$$(i, j)_N =  | ( 1+C_i^{(N,q)})   \cap C_j^{(N,q)}|$$
for all $0 \le i \le N-1$ and $0 \le j\le N-1$.

The following Lemma is proved in \cite{S}.

 \begin{lem}
 If $q\equiv 1\pmod4$, then $$(0,0)_2=\frac{q-5}{4}, (0,1)_2=(1,0)_2=(1,1)_2=\frac{q-1}{4}.$$
 If $q\equiv 3\pmod4$, then $$(0,1)_2=\frac{q+1}{4}, (0,0)_2=(1,0)_2=(1,1)_2=\frac{q-3}{4}.$$
 \end{lem}

It is well known that
\begin{eqnarray*}\mathbb{Z}/l^2\mathbb{Z}=\{0\}\cup l(\mathbb{Z}/l^2\mathbb{Z})^\ast\cup (\mathbb{Z}/l^2\mathbb{Z})^\ast.\end{eqnarray*}
For convenience, denote \begin{equation*}S= S_0\cup S_1\cup  S_{2},\end{equation*} where $S=\{i:0\le i\le l^2-1\}$, $S_0=\{0\}$,  $S_1=l(\Bbb Z/l^2\Bbb Z)^*=l(\Bbb Z/l\Bbb Z)^*=\{lu| \gcd(u,l)=1\}$, and $S_2=(\Bbb Z/l^2\Bbb Z)^*=\{u| \gcd(u,l)=1\}$. Furthermore $|S_1|=l-1$ and $|S_2|=l(l-1)$.

Let $\gamma$ be a primitive root of $(\mathbb{Z}/l^2\mathbb{Z})^\ast$, then $\gamma$ is also a primitive root of  $(\mathbb{Z}/l\mathbb{Z})^\ast$. Then $(\mathbb{Z}/l\mathbb{Z})^\ast=H_1^{(0)}\cup H_1^{(1)}$, where  $H_1^{(0)}=\langle \gamma^2 \rangle $ and $H_1^{(1)}=\gamma H_1^{(0)}$ consist of all square elements and non-square elements of $(\mathbb{Z}/l\mathbb{Z})^\ast$, respectively. Then $(\mathbb{Z}/l^2\mathbb{Z})^\ast=H_2^{(0)}\cup H_2^{(1)}$, where  $H_2^{(0)}=\langle \gamma^2 \rangle $ and $H_2^{(1)}=\gamma H_2^{(0)}$ consist of all square elements and non-square elements of $(\mathbb{Z}/l^2\mathbb{Z})^\ast$, respectively.  By \cite{HYW},
\begin{eqnarray*}&&H_2^{(0)}=\{a_0+a_1l| a_0\in H_1^{(0)}, a_1\in \mathbb{Z}/l\mathbb Z\},\nonumber\\
&&H_2^{(1)}=\{a_0+a_1l| a_0\in H_1^{(1)}, a_1\in \mathbb{Z}/l\mathbb Z\}.\end{eqnarray*}
It is clear that
$|H_2^{(0)}|=| H_2^{(1)}|=\frac{l(l-1)}2$. Hence
\begin{equation} S=S_0\cup S_1\cup S_2=S_0\cup lH_1^{(0)}\cup lH_1^{(1)}\cup H_2^{(0)}\cup H_2^{(1)}.\end{equation}

\begin{lem} Suppose that $l\equiv 3\pmod 4$, $l\ne 3$, $q=p^{\frac {l(l-1)}2}$,  $\beta=\alpha^{\frac{q-1}N}\in \Bbb F_q$, and $\ord(\beta)=N=l^2$. Then there are three cases:

(1) If $-l\not \equiv 1\pmod p$, then for $0\le i\le N-1$,  $$\Tr_{q/p}(\beta^i)=\left\{\begin{array}{ll} l(l-1)/2, &\mbox {if  $i=0$,}\\ l\varepsilon\neq 0, &\mbox{if $i\in lH_1^{(0)}$,}
\\ -l(1+\varepsilon) \neq 0, &\mbox{if $i\in lH_1^{(1)}$,}\\ 0, & \mbox{otherewise},\end{array}\right.$$ where $\varepsilon=\frac{-1+\sqrt {-l}}{2}$ and $\sqrt {-l}$ is an element in $\Bbb F_p$ such that $(\sqrt {-l})^2\equiv {-l}\pmod p$.

(2) If $-l\equiv 1\pmod p$ and $\sqrt {-l}\equiv 1\pmod{\mathcal P_1}$,  then for $0\le i\le N-1$,  \begin{eqnarray*}\Tr_{q/p}(\beta^i)=\left\{\begin{array}{ll}1, &\mbox{if $i=0$}  \mbox { or $i\in lH_1^{(1)}$,}\\ 0, & \mbox{otherewise}.\end{array}\right.\end{eqnarray*}

(3)  If $-l\equiv 1\pmod p$ and $\sqrt {-l}\equiv -1\pmod{\mathcal P_1}$,  then for $0\le i\le N-1$, $$\Tr_{q/p}(\beta^i)=\left\{\begin{array}{ll}1, &\mbox{if $i=0$}  \mbox { or $i\in lH_1^{(0)}$,}\\ 0, & \mbox{otherewise},\end{array}\right.$$
 where $\mathcal {P}_1=\langle\frac{a+b\sqrt {-l}}2, p\rangle$  is a    prime ideal  of $\Bbb Q(\zeta_l)$ over $p$, and $\zeta_l$  is a primitive $l$-th root of unity in the complex number field.
\end{lem}
\begin{proof} It is clear that $x^{l^2}-1=(x-1)\Phi_l(x)\Phi_{l^2}(x)$, where $\Phi_l(x)$ and $\Phi_{l^2}(x)$ are cyclotomic polynomials.

Since the order of $p$ modulo $l^2$ is  $\frac{l(l-1)}2$,  there is an irreducible factorization over $\Bbb F_p$:
 $$x^{l^2}-1=(x-1)\Phi_l^{(0)}(x)\Phi_l^{(1)}\Phi_{l^2}^{(0)}(x)\Phi_{l^2}^{(1)}(x),$$
  where $\beta^l=\xi_l$ is a primitive $l$-th root of unity in $\Bbb F_q$,  $\Phi_{l}^{(0)}(x)=\prod_{u\in H_{1}^{(0)}}(x-\xi_{l}^u)$, $\Phi_{l}^{(1)}(x)=\prod_{u\in H_{1}^{(1)}}(x-\xi_{l}^u)$, $ \Phi_{l^{2}}^{(0)}(x)=\Phi_{l}^{(0)}(x^l)$, and  $\Phi_{l^{2}}^{(1)}(x)=\Phi_{l}^{(1)}(x^l)$.

By Lemma 2.2, $$\sum_{u\in H_1^{(0)}}\zeta_l^u=\frac{-1+\sqrt{-l}}2, \sum_{u\in H_1^{(1)}}\zeta_l^u=\frac{-1-\sqrt{-l}}2,$$
where $\zeta_l$ is a primitive $l$-th root of unity in the complex number field.

Let $\mathcal O=\mathbb{Z}[\zeta_l]$  be the algebraic integer ring of $\mathbb{Q}(\zeta_l)$. Then there is a prime ideal factorization of the prime $p$ in the integer ring $\mathcal O$:
$$p\mathcal O=\mathcal{P}_1\mathcal{P}_2,$$
where $\mathcal {P}_1=\langle\frac{a+b\sqrt {-l}}2, p\rangle$  and $\mathcal P_2=\langle\frac{a-b\sqrt {-l}}2, p\rangle$ are   prime ideals  of $\Bbb Q(\zeta_l)$ over $p$.

  Let $$\zeta_l\equiv \xi_l=\beta^l\pmod{\mathcal{P}_1}\mbox{ and }
\mathcal O/\mathcal{P}_1=\Bbb F_{p^{\frac{l-1}{2}}}=\Bbb F_p(\xi_l),$$ where $\frac{l-1}{2}$ is the order of $p$ modulo $l$.

 Suppose that   $-l\not\equiv1\pmod p$. Then   $$  \sum_{u\in H_1^{(0)}}\zeta_l^u=\frac{-1+\sqrt {-l}}{2}\equiv \sum_{u\in H_1^{(0)}}\xi_{l}^u=\Tr_{p^{(l-1)/2}}(\beta^{lu})=\varepsilon \neq 0\pmod {{\mathcal P}_1},$$
 and  $$  \sum_{u\in H_1^{(1)}}\zeta_l^u=\frac{-1-\sqrt {-l}}{2}\equiv \sum_{u\in  H_1^{(1)}}\xi_{l}^u=\Tr_{p^{(l-1)/2}}(\beta^{lu})=-1-\varepsilon \neq 0\pmod {{\mathcal P}_1}.$$
 Hence $$\Tr_{p^{(l-1)/2}/p}(\beta^i)=
\left\{\begin{array}{ll}\varepsilon, &\mbox{ if } i=lu\in S_1, u\in H_1^{(0)}\\
-1-\varepsilon, &\mbox { if }i=lu\in S_1, u\in H_1^{(1)}.\end{array}\right.$$

 Moreover,  $\Tr_{q/p}(1)=l(l-1)/2\ne 0$; $\Tr_{q/p}(\beta^i)=0$ for $i\in (\Bbb Z/{l^2}\Bbb Z)^*$ by  $ \Phi_{l^{2}}^{(0)}(x)=\Phi_{l}^{(0)}(x^l)$ and  $\Phi_{l^{2}}^{(1)}(x)=\Phi_{l}^{(1)}(x^l)$.

Suppose that $-l\equiv 1\pmod p$ and  $\sqrt {-l}\equiv 1\pmod{\mathcal P_1}$. Then

$$\sum_{u\in H_1^{(0)}}\zeta_l^u\equiv\sum_{u\in  H_1^{(0)}}\xi_{l}^u = 0\pmod {\mathcal P_1}, \sum_{u\in H_1^{(1)}}\zeta_l^u\equiv\sum_{u\in H_1^{(1)}}\xi_{l}^u = -1\pmod {\mathcal P_1}.$$
Hence $$\Tr_{p^{(l-1)/2}/p}(\beta^i)=
\left\{\begin{array}{ll}0, &\mbox { if }i=lu\in S_1, u\in H_1^{(0)}\\
-1, &\mbox{ if } i=lu\in S_1, u\in H_1^{(1)}.\end{array}\right.$$

Moreover, $\Tr_{q/p}(1)=l(l-1)/2\equiv -l\equiv 1\pmod p$ and $\Tr_{q/p}(\beta^i)=0$ for $i\in (\Bbb Z/l^2\Bbb Z)^*$.
Note that $\Tr_{q/p}(\beta^i)=\Tr_{q/p^{(l-1)/2}}(\Tr_{p^{(l-1)/2}/p}(\beta^i))$. Then the results follow.

Suppose that $-l\equiv 1\pmod p$ and $\sqrt {-l}\equiv -1\pmod{\mathcal P_1}$.   The desire result follows from  the similar way of the  case above.

\end{proof}

Suppose that $f(x)=Tr_{q/p}(x^{\frac {q-1}N})$ is a function from $\Bbb F_q$ to $\Bbb F_p$. Then  the Walsh transform of $f(x)$ can be described by $$\widehat{f }(b)= \sum_{x\in \Bbb F_{q} } \zeta_p^{ f(x)+\Tr_{q/p}(bx)}=\sum_{x\in \Bbb F_q}\zeta_p^{\Tr_{q/p}(x^{\frac{q-1}{N}}+bx)}.$$

 If $b=0$, then $$\widehat{f }(0)=\sum_{x\in \Bbb F_q}\zeta_p^{\Tr_{q/p}(x^{\frac{q-1}{N}})}=\sum_{x\in \Bbb F_q}\psi(x^{\frac{q-1}N}).$$

If  $b\in \Bbb F_{q}^*$ and  $b^{\frac{q-1}N}=\beta^k$, where $\beta=\alpha^{\frac{q-1}N}$ and  $0\le k\le N-1$. then
\begin{eqnarray*}\widehat{f }(b)&=&\sum_{x\in \Bbb F_q}\zeta_p^{\Tr_{q/p}(x^{\frac{q-1}{N}}+bx)}=1+\sum_{x\in \Bbb F_q^*}\zeta_p^{\Tr_{q/p}(b^{-\frac{q-1}{N}}x^{\frac{q-1}{N}}+x)}=1+\sum_{x\in \Bbb F_q^*}\zeta_p^{\Tr_{q/p}(\beta^{-k}x^{\frac{q-1}{N}}+x)}.\end{eqnarray*}

 In the following, we compute the valuations  of $\widehat f(b)$, $b\in \Bbb F_q$.

\begin{lem}
$$\widehat{f }(0)=\frac{l-1}2p^{\frac{f-hl}2}((\frac{a+b\sqrt{-l}}2)^{l}+(\frac{a-b\sqrt{-l}}2)^{l})+\frac{l(l-1)}2p^{\frac{f-h}2}a,$$
where $h$ is the ideal class number of $\mathbb{Q}(\sqrt{-l})$, and $a,b$ are integers given by (2.3).
\end{lem}

\begin{proof} Let $\psi(x)=\zeta_p^{Tr_{q/p}(x)}$ be a canonical additive character of $\Bbb F_q$.
By Lemma 2.1,
$$\widehat f(0)=\sum_{x\in \Bbb F_q}\psi(x^{\frac{q-1}N})=\sum_{j=1}^{N-1}G(\chi^{j}), $$
where $\chi$ is a multiplicative character of order $N$.

By $j\in lH_1^{(0)}\cup l H_1^{(1)}\cup H_2^{(0)}\cup H_2^{(1)}$ and Lemma 2.3,
\begin{eqnarray*}\widehat{f }(0)&=&\sum_{j\in lH_1^{(0)}}G(\chi^j)
+\sum_{j\in lH_1^{(1)}}G(\chi^j)+\sum_{j\in H_2^{(0)}}G(\chi^j)+\sum_{j\in H_2^{(1)}}G(\chi^j)
\\&=&\sum_{v\in H_1^{(0)}}G(\chi^{lv})
+\sum_{v\in H_1^{(1)}}G(\chi^{lv})+\sum_{j\in H_2^{(0)}}G(\chi^j)+\sum_{j\in H_2^{(1)}}G(\chi^j)\\
&=&\sum_{v\in H_1^{(0)}}G(\chi^{l})
+\sum_{v\in H_1^{(1)}}G(\chi^{l})+\sum_{j\in H_2^{(0)}}G(\chi)+\sum_{j\in H_2^{(1)}}G(\chi)\\
&=&\frac{l-1}2p^{\frac{f-hl}2}((\frac{a+b\sqrt{-l}}2)^{l}+(\frac{a-b\sqrt{-l}}2)^{l})+\frac{l(l-1)}2p^{\frac{f-h}2}a.
\end{eqnarray*}
\end{proof}

Let $\Bbb F_q^*=\langle \alpha\rangle $.  By the divisor  algorithm,  for an integer $s$ with $0\le s\le q-2$,  $$s=jN+i, 0\le j\le \frac{q-1}N-1, 0\le i\le N-1.$$
For $b\ne 0$, $\beta=\alpha^{\frac{q-1}N}$,  and $b^{\frac{q-1}N}=\beta^k$,
\begin{eqnarray*}\widehat f(b)&=&1+\sum_{s=0}^{q-2}\psi(\beta^{-k}\alpha^{\frac{q-1}N s}+\alpha^s)=1+\sum_{i=0}^{N-1}\psi(\beta^{i-k})\sum_{j=0}^{\frac{q-1}N-1}\psi(\alpha^i\alpha^{Nj})\\ &=&1+\frac 1N\sum_{i=0}^{N-1}\psi(\beta^{i-k})\sum_{x\in \Bbb F_q^*}\psi(\alpha^ix^{N})=1+\frac 1N\sum_{i=0}^{N-1}\psi(\beta^{i-k})(\sum_{x\in \Bbb F_q}\psi(\alpha^ix^{N})-1)
\\ &=&1+\frac 1N\sum_{i=0}^{N-1}\psi(\beta^{i-k})\sum_{j=0}^{N-1}\overline{\chi^j(\alpha^i})G(\chi^j)=1+\frac 1N\sum_{i=0}^{N-1}\psi(\beta^{i-k})\sum_{j=0}^{N-1}\zeta_N^{-ij}G(\chi^j),\end{eqnarray*}
where $\zeta_N=e^{\frac{2\pi\sqrt{-1}}N}$ is the $N$-th primitive root of unity in the complex field, $\chi$ is a primitive multiplicative character of order $N$ over $\Bbb F_q^*$, $\chi(\alpha)=\zeta_N$, and $G(\chi^0)=-1$.
By Lemma 2.3,
\begin{eqnarray*}\sum_{j=0}^{N-1}\zeta_N^{-ij}G(\chi^j)&=&-1+p^{\frac{f-hl}2}((\frac{a+b\sqrt{-l}}2)^l\sum_{j\in lH_1^{(0)}}\zeta_N^{-ij}+(\frac{a-b\sqrt{-l}}2)^l\sum_{j\in lH_1^{(1)}}\zeta_N^{-ij})\\ &+&p^{\frac{f-h}2}(\frac{a+b\sqrt {-l}}2\sum_{j\in H_2^{(0)}}\zeta_N^{-ij}+\frac{a-b\sqrt {-l}}2\sum_{j\in H_2^{(1)}}\zeta_N^{-ij}). \end{eqnarray*}
Hence
\begin{eqnarray}\widehat f(b)&=&1+\frac{-I_0}N+\frac{p^{\frac{f-hl}2}}N((\frac{a+b\sqrt{-l}}2)^lI_1^{(0)}+(\frac{a-b\sqrt{-l}}2)^lI_1^{(1)})\nonumber\\
&+& \frac{p^{\frac{f-h}2}}N(\frac{a+b\sqrt{-l}}2I_2^{(0)}+\frac{a-b\sqrt{-l}}2I_2^{(1)}),\end{eqnarray}
where $$I_0=\sum_{i=0}^{N-1}\psi(\beta^{i-k})=\sum_{i=0}^{N-1}\psi(\beta^{i}),$$ $$I_1^{(0)}=\sum_{i=0}^{N-1}\psi(\beta^{i-k})\sum_{j\in lH_1^{(0)}}\zeta_N^{-ij}, I_1^{(1)}=\sum_{i=0}^{N-1}\psi(\beta^{i-k})\sum_{j\in lH_1^{(1)}}\zeta_N^{-ij},$$
$$I_2^{(0)}=\sum_{i=0}^{N-1}\psi(\beta^{i-k})\sum_{j\in H_2^{(0)}}\zeta_N^{-ij}, I_2^{(1)}=\sum_{i=0}^{N-1}\psi(\beta^{i-k})\sum_{j\in H_2^{(1)}}\zeta_N^{-ij}.$$

In the following, we shall compute the values of $\widehat f(b),b\neq 0$, in two cases: $-l\not\equiv 1\pmod p$ and $-l\equiv 1\pmod p$.

\begin{lem} Suppose that $-l\not\equiv 1\pmod p$.  Let $b\ne 0$ and $b^{\frac{q-1}N}=\beta^k$, $0\le k\le N-1$. Then
$$I_0=\zeta_p^{l(l-1)/2}+\frac{l-1}2\zeta_p^{l\varepsilon}+\frac{l-1}2\zeta_p^{-l(1+\varepsilon)}+{l(l-1)}.$$
(1) If $k=0$, then
$$I_1^{(0)}=I_1^{(1)}=\frac{l-1}2\zeta_p^{l(l-1)/2}+(\frac{l-1}2)^2\zeta_p^{l\varepsilon}+(\frac{l-1}2)^2\zeta_p^{-l(1+\varepsilon)}-\frac{l(l-1)}2.$$
$$I_2^{(0)}=\frac{l(l-1)}2\zeta_p^{l(l-1)/2}+\frac{l(l-1)}2\frac{-1-\sqrt{-l}}2\zeta_p^{l\varepsilon }+\frac{l(l-1)}2\frac{-1+\sqrt{-l}}2\zeta_p^{-l(1+\varepsilon)}.$$
$$I_2^{(1)}=\frac{l(l-1)}2\zeta_p^{l(l-1)/2}+\frac{l(l-1)}2\frac{-1+\sqrt{-l}}2\zeta_p^{l\varepsilon} +\frac{l(l-1)}2\frac{-1-\sqrt{-l}}2\zeta_p^{-l(1+\varepsilon)}.$$
(2) If $k\in lH_1^{(0)}$, then
$$I_1^{(0)}=I_1^{(1)}=\frac{l-1}2\zeta_p^{l(l-1)/2}+(\frac{l-1}2)^2\zeta_p^{l\varepsilon}+(\frac{l-1}2)^2\zeta_p^{-l(1+\varepsilon)}-\frac{l(l-1)}2.$$
$$I_2^{(0)}=l\frac{-1-\sqrt{-l}}2\zeta_p^{l(l-1)/2}+\frac{l^2+l}4\zeta_p^{-l(1+\varepsilon)}+l\zeta_p^{l\varepsilon} (\frac{1+\sqrt{-l}}2)^2.$$
$$I_2^{(1)}=l\frac{-1+\sqrt{-l}}2\zeta_p^{l(l-1)/2}+\frac{l^2+l}4\zeta_p^{-l(1+\varepsilon)}+l\zeta_p^{l\varepsilon} (\frac{1-\sqrt{-l}}2)^2.$$
(3) If $k\in lH_1^{(1)}$, then
$$I_1^{(0)}=I_1^{(1)}=\frac{l-1}2\zeta_p^{l(l-1)/2}+(\frac{l-1}2)^2\zeta_p^{l\varepsilon}+(\frac{l-1}2)^2\zeta_p^{-l(1+\varepsilon)}-\frac{l(l-1)}2. $$
$$I_2^{(0)}=l\frac{-1+\sqrt{-l}}2\zeta_p^{l(l-1)/2}+\frac{l^2+l}4\zeta_p^{l\varepsilon}+l\zeta_p^{-l(1+\varepsilon)} (\frac{1-\sqrt{-l}}2)^2.$$
$$I_2^{(1)}=l\frac{-1-\sqrt{-l}}2\zeta_p^{l(l-1)/2}+\frac{l^2+l}4\zeta_p^{l\varepsilon}+l\zeta_p^{-l(1+\varepsilon)} (\frac{1+\sqrt{-l}}2)^2.$$
(4) If $k\in H_2^{(0)}$, then
$$I_1^{(0)}=l\frac{1+\sqrt{-l}}2+\frac{(l-1)(-1-\sqrt{-l})}4(\zeta_p^{l\varepsilon}+\zeta_p^{-l(1+\varepsilon)})+\frac{-1-\sqrt{-l}}2\zeta_p^{l(l-1)/2}.$$
$$I_1^{(1)}=l\frac{1-\sqrt{-l}}2+\frac{(l-1)(-1+\sqrt{-l})}4(\zeta_p^{l\varepsilon}+\zeta_p^{-l(1+\varepsilon)})+\frac{-1+\sqrt{-l}}2\zeta_p^{l(l-1)/2}.$$
$$I_2^{(0)}=I_2^{(1)}=0.$$

(5) If $k\in H_2^{(1)}$, then
$$I_1^{(0)}=l\frac{1-\sqrt{-l}}2+\frac{(l-1)(-1+\sqrt{-l})}4(\zeta_p^{l\varepsilon}+\zeta_p^{-l(1+\varepsilon)})+\frac{-1+\sqrt{-l}}2\zeta_p^{l(l-1)/2}.$$
$$I_1^{(1)}=l\frac{1+\sqrt{-l}}2+\frac{(l-1)(-1-\sqrt{-l})}4(\zeta_p^{l\varepsilon}+\zeta_p^{-l(1+\varepsilon)})+\frac{-1-\sqrt{-l}}2\zeta_p^{l(l-1)/2}.$$
$$I_2^{(0)}=I_2^{(1)}=0.$$
\end{lem}
\begin{proof} Note that
$$I_0=\sum_{i=0}^{N-1}\psi(\beta^{i-k})=\sum_{i=0}^{N-1}\psi(\beta^{i}).$$   By Lemma 3.2, $$I_0=\zeta_p^{l(l-1)/2}+\frac{l-1}2\zeta_p^{l\varepsilon}+\frac{l-1}2\zeta_p^{-l(1+\varepsilon)}+{l(l-1)}.$$

(1) Suppose that $k=0$. Then
\begin{eqnarray*}I_1^{(0)}=\sum_{i\in l(\Bbb Z/l^2\Bbb Z)}\psi(\beta^{i})\frac{l-1}2+\sum_{i\in (\Bbb Z/l^2\Bbb Z)^*}\psi(\beta^{i})\sum_{j\in lH_1^{(0)}}\zeta_N^{-ij}.\end{eqnarray*}

If  $j\in lH_1^{(0)}$ and $i\in H_2^{(0)}$,  then $-ij\in lH_1^{(1)}$ by $l\equiv 3\pmod 4$ and $\sum_{j\in lH_1^{(0)}}\zeta_N^{-ij}=\frac{-1-\sqrt{-l}}2$.
If $j\in lH_1^{(0)}$ and $i\in H_2^{(1)}$, then  $\sum_{j\in lH_1^{(0)}}\zeta_N^{-ij}=\frac{-1+\sqrt{-l}}2$.
Hence
\begin{eqnarray*}I_1^{(0)}&=&\frac{l-1}2\zeta_p^{l(l-1)/2}+(\frac{l-1}2)^2\zeta_p^{l\varepsilon}+(\frac{l-1}2)^2\zeta_p^{-l(1+\varepsilon)}+
\frac{l(l-1)}2(\frac{-1-\sqrt{-l}}2+\frac{-1+\sqrt{-l}}2)
\\ &=&\frac{l-1}2\zeta_p^{l(l-1)/2}+(\frac{l-1}2)^2\zeta_p^{l\varepsilon}+(\frac{l-1}2)^2\zeta_p^{-l(1+\varepsilon)}-\frac{l(l-1)}2.
\end{eqnarray*}
Similarly, $I_1^{(1)}=I_1^{(0)}. $

 If $i\in (\Bbb Z/l^2\Bbb Z)^*$, then by $\Phi_{l^2}^{(0)}(x)=\Phi_l^{(0)}(x^l)$ and $\Phi_{l^2}^{(1)}(x)=\Phi_l^{(1)}(x^l)$, $\sum_{j\in H_2^{(0)}}\zeta_N^{-ij}=0$. Hence
 $$I_2^{(0)}=\sum_{i\in l(\Bbb Z/l^2\Bbb Z)^*}\psi(\beta^i)\sum_{j\in H_2^{(0)}}\zeta_N^{-ij}.$$

If $j\in H_2^{(0)}$ and $i\in lH_1^{(0)}$, and  $i=i_0l, j=j_0+j_1l$,  where $i_0, j_0\in H_1^{(0)}, j_1\in \mathbb{Z}/l\mathbb Z$,
 then  $\zeta_N^{-ij}=\zeta_l^{-i_0j_0}$ and $-i_0j_0\in H_1^{(1)}$  by $l\equiv 3\pmod 4$.
Note  that $j_1\in \{0,1,\ldots l-1\}$. Then  the valuation of  $\sum_{ j\in H_2^{(0)}} \zeta_{N}^{-ij}$ is  exactly $l$ times valuation of    $\sum_{ j_0\in H_1^{(0)}} \zeta_{l}^{-i_0j_0}$.
Hence  $\sum_{j\in H_2^{(0)}}\zeta_N^{-ij}=l\frac{-1-\sqrt{-l}}2$; if $j\in H_2^{(0)}$ and $i\in lH_1^{(1)}$, then $\sum_{j\in H_2^{(0)}}\zeta_N^{-ij}=l\frac{-1+\sqrt{-l}}2$;
Hence
$$I_2^{(0)}=\frac{l(l-1)}2\zeta_p^{l(l-1)/2}+\frac{l(l-1)}2\frac{-1-\sqrt{-l}}2\zeta_p^{l\varepsilon }+\frac{l(l-1)}2\frac{-1+\sqrt{-l}}2\zeta_p^{-l(1+\varepsilon)}.$$
Similarly, if $j\in H_2^{(1)}$ and $i\in lH_1^{(0)}$, then $\sum_{j\in H_2^{(1)}}\zeta_N^{-ij}=l\frac{-1+\sqrt{-l}}2$.  Hence
$$I_2^{(1)}=\frac{l(l-1)}2\zeta_p^{l(l-1)/2}+\frac{l(l-1)}2\frac{-1+\sqrt{-l}}2\zeta_p^{l\varepsilon} +\frac{l(l-1)}2\frac{-1-\sqrt{-l}}2\zeta_p^{-l(1+\varepsilon)}.$$

(2)  Suppose that $k\in lH_1^{(0)}$. Then
\begin{eqnarray*}I_1^{(0)}&=&\sum_{i\in l(\Bbb Z/l^2\Bbb Z)}\psi(\beta^{i-k})\frac{l-1}2+\sum_{i\in (\Bbb Z/l^2\Bbb Z)^*}\psi(\beta^{i-k})\sum_{j\in lH_1^{(0)}}\zeta_N^{-ij}\\&=&\sum_{i\in l(\Bbb Z/l^2\Bbb Z)}\psi(\beta^{i})\frac{l-1}2+\sum_{i\in (\Bbb Z/l^2\Bbb Z)^*}\psi(\beta^{i})\sum_{j\in lH_1^{(0)}}\zeta_N^{-ij}\\ &=&\frac{l-1}2\zeta_p^{l(l-1)/2}+(\frac{l-1}2)^2\zeta_p^{l\varepsilon}+(\frac{l-1}2)^2\zeta_p^{-l(1+\varepsilon)}-\frac{l(l-1)}2.
\end{eqnarray*}
Similarly, $I_1^{(1)}=I_1^{(0)}. $

  By $k\in lH_1^{(0)}$,  there is a unique $i\in lH_1^{(0)}$ such that $i=k$; there are $(1,1)_2=\frac{l-3}4$ elements $i\in lH_1^{(0)}$ such that $i-k\in lH_1^{(0)}$; there are $(1,0)_2=\frac{l-3}4$ elements $i\in lH_1^{(0)}$ such that $i-k\in lH_1^{(1)}$. On the other hand, there are $(0,1)_2=\frac{l+1}4$ elements $i\in lH_1^{(1)}$ such that $i-k\in lH_1^{(0)}$; there are $(0,0)_2=\frac{l-3}4$ elements $i\in lH_1^{(1)}$ such that $i-k\in lH_1^{(1)}$. Thus
\begin{eqnarray*}I_2^{(0)}&=&\sum_{i\in l(\Bbb Z/l^2\Bbb Z)}\psi(\beta^{i-k})\sum_{j\in H_2^{(0)}}\zeta_N^{-ij}\\
&=&\psi(\beta^{0-k})\sum_{j\in H_2^{(0)}}\zeta_N^{0}+\sum_{i\in lH_1^{(0)}}\psi(\beta^{i-k})\sum_{j\in H_2^{(0)}}\zeta_N^{-ij}+\sum_{i\in lH_1^{(1)}}\psi(\beta^{i-k})\sum_{j\in H_2^{(0)}}\zeta_N^{-ij}\\
&=&\zeta_p^{-l(1+\varepsilon)}\frac{l(l-1)}2+l\frac{-1-\sqrt{-l}}2\zeta_p^{l(l-1)/2}+\frac{l-3}4\zeta_p^{l\varepsilon}l\frac{-1-\sqrt{-l}}2\\ &+&\frac{l-3}4\zeta_p^{-l(1+\varepsilon)}l\frac{-1-\sqrt{-l}}2+\frac{l+1}4\zeta_p^{l\varepsilon }l\frac{-1+\sqrt{-l}}2+\frac{l-3}4\zeta_p^{-l(1+\varepsilon)}l\frac{-1+\sqrt{-l}}2\\
&=&l\frac{-1-\sqrt{-l}}2\zeta_p^{l(l-1)/2}+\frac{l^2+l}4\zeta_p^{-l(1+\varepsilon)}+l\zeta_p^{l\varepsilon} (\frac{1+\sqrt{-l}}2)^2,
\end{eqnarray*}
\begin{eqnarray*}\mbox{and }I_2^{(1)}&=&\sum_{i\in l(\Bbb Z/l^2\Bbb Z)}\psi(\beta^{i-k})\sum_{j\in H_2^{(1)}}\zeta_N^{-ij}\\
&=&\psi(\beta^{0-k})\sum_{j\in H_2^{(1)}}\zeta_N^{0}+\sum_{i\in lH_1^{(0)}}\psi(\beta^{i-k})\sum_{j\in H_2^{(1)}}\zeta_N^{-ij}+\sum_{i\in lH_1^{(1)}}\psi(\beta^{i-k})\sum_{j\in H_2^{(1)}}\zeta_N^{-ij}\\
&=&\zeta_p^{-l(1+\varepsilon)}\frac{l(l-1)}2+l\frac{-1+\sqrt{-l}}2\zeta_p^{l(l-1)/2}+\frac{l-3}4\zeta_p^{l\varepsilon}l\frac{-1+\sqrt{-l}}2\\ &+&\frac{l-3}4\zeta_p^{-l(1+\varepsilon)}l\frac{-1+\sqrt{-l}}2+\frac{l+1}4\zeta_p^{l\varepsilon }l\frac{-1-\sqrt{-l}}2+\frac{l-3}4\zeta_p^{-l(1+\varepsilon)}l\frac{-l-\sqrt{-l}}2\\
&=&l\frac{-1+\sqrt{-l}}2\zeta_p^{l(l-1)/2}+\frac{l^2+l}4\zeta_p^{-l(1+\varepsilon)}+l\zeta_p^{l\varepsilon} (\frac{1-\sqrt{-l}}2)^2.
\end{eqnarray*}

(3) Suppose that $k\in lH_1^{(1)}$. Then
$$I_1^{(0)}=I_1^{(1)}=\frac{l-1}2\zeta_p^{l(l-1)/2}+(\frac{l-1}2)^2\zeta_p^{l\varepsilon}+(\frac{l-1}2)^2\zeta_p^{-l(1+\varepsilon)}-\frac{l(l-1)}2.$$

  By $k\in lH_1^{(1)}$,  there are $(0,0)_2=\frac{l-3}4$ elements $i\in lH_1^{(0)}$ such that $i-k\in lH_1^{(0)}$; there are $(0,1)_2=\frac{l+1}4$ elements $i\in lH_1^{(0)}$ such that $i-k\in lH_1^{(1)}$. On the other hand, there is a unique  $i\in lH_1^{(1)}$ such that $i=k$; there are $(1,0)_2=\frac{l-3}4$ elements $i\in lH_1^{(1)}$ such that $i-k\in lH_l^{(0)}$; there are $(1,1)_2=\frac{l-3}4$ elements $i\in lH_1^{(1)}$ such that $i-k\in lH_l^{(1)}$.
Thus
\begin{eqnarray*}I_2^{(0)}&=&\sum_{i\in l(\Bbb Z/l^2\Bbb Z)}\psi(\beta^{i-k})\sum_{j\in H_2^{(0)}}\zeta_N^{-ij}\\
&=&\psi(\beta^{0-k})\sum_{j\in H_2^{(0)}}\zeta_N^{0}+\sum_{i\in lH_1^{(0)}}\psi(\beta^{i-k})\sum_{j\in H_2^{(0)}}\zeta_N^{-ij}+\sum_{i\in lH_1^{(1)}}\psi(\beta^{i-k})\sum_{j\in H_2^{(0)}}\zeta_N^{-ij}\\
&=&\zeta_p^{l\varepsilon}\frac{l(l-1)}2+\frac{l-3}4\zeta_p^{l\varepsilon}l\frac{-1-\sqrt{-l}}2+\frac{l+1}4\zeta_p^{-l(1+\varepsilon)}l\frac{-1-\sqrt{-l}}2\\ &+&l\frac{-1+\sqrt{-l}}2\zeta_p^{l(l-1)/2}+\frac{l-3}4\zeta_p^{l\varepsilon }l\frac{-1+\sqrt{-l}}2+\frac{l-3}4\zeta_p^{-l(1+\varepsilon)}l\frac{-1+\sqrt{-l}}2\\
&=&l\frac{-1+\sqrt{-l}}2\zeta_p^{l(l-1)/2}+\frac{l^2+l}4\zeta_p^{l\varepsilon}+l\zeta_p^{-l(1+\varepsilon)} (\frac{1-\sqrt{-l}}2)^2.
\end{eqnarray*}
\begin{eqnarray*} \mbox{and }I_2^{(1)}&=&\sum_{i\in l(\Bbb Z/l^2\Bbb Z)}\psi(\beta^{i-k})\sum_{j\in H_2^{(1)}}\zeta_N^{-ij}\\
&=&\psi(\beta^{0-k})\sum_{j\in H_2^{(1)}}\zeta_N^{0}+\sum_{i\in lH_1^{(0)}}\psi(\beta^{i-k})\sum_{j\in H_2^{(1)}}\zeta_N^{-ij}+\sum_{i\in lH_1^{(1)}}\psi(\beta^{i-k})\sum_{j\in H_2^{(1)}}\zeta_N^{-ij}\\
&=&\zeta_p^{l\varepsilon}\frac{l(l-1)}2+\frac{l-3}4\zeta_p^{l\varepsilon}l\frac{-1+\sqrt{-l}}2+\frac{l+1}4\zeta_p^{-l(1+\varepsilon)}l\frac{-1+\sqrt{-l}}2\\ &+&l\frac{-1-\sqrt{-l}}2\zeta_p^{l(l-1)/2}+\frac{l-3}4\zeta_p^{l\varepsilon }l\frac{-1-\sqrt{-l}}2+\frac{l-3}4\zeta_p^{-l(1+\varepsilon)}l\frac{-1-\sqrt{-l}}2\\
&=&l\frac{-1-\sqrt{-l}}2\zeta_p^{l(l-1)/2}+\frac{l^2+l}4\zeta_p^{l\varepsilon}+l\zeta_p^{-l(1+\varepsilon)} (\frac{1+\sqrt{-l}}2)^2.
\end{eqnarray*}

(4) Suppose that $k\in H_2^{(0)}$. Then there is a unique $i\in H_2^{(0)}$ such that $i=k$; there are $\frac{l-1}2$ elements $i\in H_2^{(0)}$ such that $i-k\in lH_1^{(0)}$; there are $\frac{l-1}2 $ elements $i\in H_2^{(0)}$ such that $i-k\in lH_1^{(1)}$; there are $\frac{l(l-1)}2-l$ elements $i\in H_2^{(0)}$ such that $i-k\in (\Bbb Z/ l^2\Bbb Z)^*$. On the other hand, there are $\frac{l(l-1)}2$  elements $i\in H_2^{(1)}$ such that $i-k\in (\Bbb Z/l^2\Bbb Z)^*$. Then
\begin{eqnarray*}I_1^{(0)}&=&\sum_{i\in l(\Bbb Z/l^2\Bbb Z)}\psi(\beta^{i-k})\frac{l-1}2+\sum_{i\in (\Bbb Z/l^2\Bbb Z)^*}\psi(\beta^{i-k})\sum_{j\in lH_1^{(0)}}\zeta_N^{-ij}\\
&=& \frac {l(l-1)}2+\zeta_p^{l(l-1)/2}\frac{-1-\sqrt{-l}}2+\frac{l-1}2\zeta_p^{l\varepsilon}\frac{-1-\sqrt{-l}}2\\&+&\frac{l-1}2\zeta_p^{-l(1+\varepsilon)}\frac{-1-\sqrt{-l}}2
+\frac{l(l-3)}2\frac{-1-\sqrt{-l}}2+\frac{l(l-1)}2\frac{-1+\sqrt{-l}}2\\
&=&l\frac{1+\sqrt{-l}}2+\frac{(l-1)(-1-\sqrt{-l})}4(\zeta_p^{l\varepsilon}+\zeta_p^{-l(1+\varepsilon)})+\frac{-1-\sqrt{-l}}2\zeta_p^{l(l-1)/2},\end{eqnarray*}
\begin{eqnarray*}I_1^{(1)}&=&\sum_{i\in l(\Bbb Z/l^2\Bbb Z)}\psi(\beta^{i-k})\frac{l-1}2+\sum_{i\in (\Bbb Z/l^2\Bbb Z)^*}\psi(\beta^{i-k})\sum_{j\in lH_1^{(1)}}\zeta_N^{-ij}\\
&=& \frac {l(l-1)}2+\zeta_p^{l(l-1)/2}\frac{-1+\sqrt{-l}}2+\frac{l-1}2\zeta_p^{l\varepsilon}\frac{-1+\sqrt{-l}}2\\ &+&\frac{l-1}2\zeta_p^{-l(1+\varepsilon)}\frac{-1+\sqrt{-l}}2
+\frac{l(l-3)}2\frac{-1+\sqrt{-l}}2+\frac{l(l-1)}2\frac{-1-\sqrt{-l}}2\\
&=&l\frac{1-\sqrt{-l}}2+\frac{(l-1)(-1+\sqrt{-l})}4(\zeta_p^{l\varepsilon}+\zeta_p^{-l(1+\varepsilon)})+\frac{-1+\sqrt{-l}}2\zeta_p^{l(l-1)/2}.
\end{eqnarray*}

By $k\in H_2^{(0)}$. If $i\in (\Bbb Z/l^2\Bbb Z)^*$,  then  $\sum_{j\in H_2^{(0)}}\zeta_N^{-ij}=0$. If $i\in l(\Bbb Z/l^2\Bbb Z)$, then $i-k \in (\Bbb Z/l^2\Bbb Z)^*$.
Hence
\begin{eqnarray}I_2^{(0)}&=&\sum_{i\in l(\Bbb Z/l^2\Bbb Z)}\psi(\beta^{i-k})\sum_{j\in H_2^{(0)}}\zeta_N^{-ij}=\sum_{i\in l(\Bbb Z/l^2\Bbb Z)}\sum_{j\in H_2^{(0)}}\zeta_N^{-ij}\\\nonumber
 &=&\frac {l(l-1)}2+\sum_{i\in lH_1^{(0)}}\sum_{j\in H_2^{(0)}}\zeta_N^{-ij}+\sum_{i\in lH_1^{(1)}}\sum_{j\in H_2^{(0)}}\zeta_N^{-ij}\\\nonumber
 &=&\frac {l(l-1)}2+l\frac{-1-\sqrt{-l}}2\frac{l-1}2+l\frac{-1+\sqrt{-l}}2\frac{l-1}2=0.
\end{eqnarray}
Similarly,  $I_2^{(1)}=0$.

(5) Suppose that $k\in H_2^{(1)}$. Then there is a unique $i\in H_2^{(1)}$ such that $i=k$; there are $\frac{l-1}2$ elements $i\in H_2^{(1)}$ such that $i-k\in lH_1^{(0)}$; there are $\frac{l-1}2 $ elements $i\in H_2^{(1)}$ such that $i-k\in lH_1^{(1)}$; there are $\frac{l(l-1)}2-l$ elements $i\in H_2^{(1)}$ such that $i-k\in (\Bbb Z/ l^2\Bbb Z)^*$. On the other hand there are $\frac{l(l-1)}2$  elements $i\in H_2^{(0)}$ such that $i-k\in (\Bbb Z/l^2\Bbb Z)^*$.
Hence
\begin{eqnarray*}I_1^{(0)}&=&\sum_{i\in l(\Bbb Z/l^2\Bbb Z)}\psi(\beta^{i-k})\frac{l-1}2+\sum_{i\in (\Bbb Z/l^2\Bbb Z)^*}\psi(\beta^{i-k})\sum_{j\in lH_1^{(0)}}\zeta_N^{-ij}\\
&=& \frac {l(l-1)}2+\zeta_p^{l(l-1)/2}\frac{-1+\sqrt{-l}}2+\frac{l-1}2\zeta_p^{-l(1+\varepsilon)}\frac{-1+\sqrt{-l}}2+\frac{l-1}2\zeta_p^{l\varepsilon}\frac{-1+\sqrt{-l}}2
\\&+&\frac{l(l-3)}2\frac{-1+\sqrt{-l}}2+\frac{l(l-1)}2\frac{-1-\sqrt{-l}}2\\
&=&l\frac{1-\sqrt{-l}}2+\frac{(l-1)(-1+\sqrt{-l})}4(\zeta_p^{l\varepsilon}+\zeta_p^{-l(1+\varepsilon)})+\frac{-1+\sqrt{-l}}2\zeta_p^{l(l-1)/2},
\end{eqnarray*}
\begin{eqnarray*}I_1^{(1)}&=&\sum_{i\in l(\Bbb Z/l^2\Bbb Z)}\psi(\beta^{i-k})\frac{l-1}2+\sum_{i\in (\Bbb Z/l^2\Bbb Z)^*}\psi(\beta^{i-k})\sum_{j\in lH_1^{(1)}}\zeta_N^{-ij}\\
&=& \frac {l(l-1)}2+\zeta_p^{l(l-1)/2}\frac{-1-\sqrt{-l}}2+\frac{l-1}2\zeta_p^{-l(1+\varepsilon)}\frac{-1-\sqrt{-l}}2\\ &+&\frac{l-1}2\zeta_p^{l\varepsilon}\frac{-1-\sqrt{-l}}2
+\frac{l(l-3)}2\frac{-1-\sqrt{-l}}2+\frac{l(l-1)}2\frac{-1+\sqrt{-l}}2\\
&=&l\frac{1+\sqrt{-l}}2+\frac{(l-1)(-1-\sqrt{-l})}4(\zeta_p^{l\varepsilon}+\zeta_p^{-l(1+\varepsilon)})+\frac{-1-\sqrt{-l}}2\zeta_p^{l(l-1)/2}.
\end{eqnarray*}

Similarly,
$I_2^{(0)}=  I_2^{(1)}=0.$

This completes the proof.
\end{proof}

Suppose that  $-l \equiv 1\pmod p$.  Similar to the proof of Lemma 3.4, the following result follows and the proof is omitted here.

\begin{lem} Suppose that $-l{\equiv} 1\pmod p$  and $\sqrt {-l}\equiv 1\pmod{\mathcal P_1}$.  Let $b\ne 0$ and $b^{\frac{q-1}N}=\beta^k$, $0\le k\le N-1$. Then
$$I_0=\frac{l+1}2(\zeta_p-1)+l^2.$$
(1) If $k=0$, then
$$I_1^{(0)}=I_1^{(1)}=\frac{l^2-1}4(\zeta_p-1),I_2^{(0)}=\frac{l(l-1)}2\frac{1+\sqrt{-l}}2(\zeta_p-1),$$
$$I_2^{(1)}=\frac{l(l-1)}2\frac{1-\sqrt{-l}}2(\zeta_p-1).$$
(2) If $k\in lH_1^{(0)}$, then
$$I_1^{(0)}=I_1^{(1)}=\frac{l^2-1}4(\zeta_p-1),I_2^{(0)}=l(\frac{1+\sqrt{-l}}2)^2(1-\zeta_p),$$
$$I_2^{(1)}=l(\frac{1-\sqrt{-l}}2)^2(1-\zeta_p).$$
(3) If $k\in lH_1^{(1)}$, then
$$I_1^{(0)}=I_1^{(1)}=\frac{l^2-1}4(\zeta_p-1),I_2^{(0)}=I_2^{(1)}=\frac{l^2+l}4(1-\zeta_p),$$
(4) If $k\in H_2^{(0)}$, then
$$I_1^{(0)}=\frac{(l+1)(1+\sqrt{-l})}4(1-\zeta_p),I_1^{(1)}=\frac{(l+1)(1-\sqrt{-l})}4(1-\zeta_p), I_2^{(0)}=I_2^{(1)}=0.$$
(5) If $k\in H_2^{(1)}$, then
$$I_1^{(0)}=\frac{(l+1)(1-\sqrt{-l})}4(1-\zeta_p),I_1^{(1)}=\frac{(l+1)(1+\sqrt{-l})}4(1-\zeta_p),I_2^{(0)}=I_2^{(1)}=0.$$
\end{lem}

\begin{lem} Suppose that $-l{\equiv} 1\pmod p$  and $\sqrt {-l}\equiv -1\pmod{\mathcal P_1}$.  Let $b\ne 0$ and $b^{\frac{q-1}N}=\beta^k$, $0\le k\le N-1$. Then
$$I_0=\frac{l+1}2(\zeta_p-1)+l^2.$$
(1) If $k=0$, then
$$I_1^{(0)}=I_1^{(1)}=\frac{l^2-1}4(\zeta_p-1),I_2^{(0)}=\frac{l(l-1)}2\frac{1-\sqrt{-l}}2(\zeta_p-1),$$
$$I_2^{(1)}=\frac{l(l-1)}2\frac{1+\sqrt{-l}}2(\zeta_p-1).$$
(2) If $k\in lH_1^{(0)}$, then
$$I_1^{(0)}=I_1^{(1)}=\frac{l^2-1}4(\zeta_p-1),I_2^{(0)}=I_2^{(1)}=\frac{l^2+l}4(1-\zeta_p).$$
(3) If $k\in lH_1^{(1)}$, then
$$I_1^{(0)}=I_1^{(1)}=\frac{l^2-1}4(\zeta_p-1),I_2^{(0)}=l (\frac{1-\sqrt{-l}}2)^2(1-\zeta_p),$$
$$I_2^{(1)}=l (\frac{1+\sqrt{-l}}2)^2(1-\zeta_p).$$
(4) If $k\in H_2^{(0)}$, then
$$I_1^{(0)}=\frac{(l+1)(1+\sqrt{-l})}4(1-\zeta_p),I_1^{(1)}=\frac{(l+1)(1-\sqrt{-l})}4(1-\zeta_p),I_2^{(0)}=I_2^{(1)}=0.$$
(5) If $k\in H_2^{(1)}$, then
$$I_1^{(0)}=\frac{(l+1)(1-\sqrt{-l})}4(1-\zeta_p),I_1^{(1)}=\frac{(l+1)(1+\sqrt{-l})}4(1-\zeta_p),I_2^{(0)}=I_2^{(1)}=0.$$
\end{lem}
By (3.2) and Lemma 3.4,  we have the following  theorem.

\begin{thm} Suppose that $ l\equiv 3\pmod 4$, $l\ne 3$,  be a prime. Let  $N=l^2$ be an integer and $f=\frac{l(l-1)}2$ be   the smallest positive integer such that $p^f\equiv 1 \pmod N$. Let $q=p^{f}$  for $\gcd(p,f)=1$.  Suppose that  $-l{\not\equiv} 1\pmod p$. Then the Walsh spectrum of    $f(x)=\Tr_{q/p}(x^{\frac{q-1}{N}})$ is given as follows:
\begin{eqnarray*}
Value&& Frequency\\
\frac{l-1}2(A+B+lp^{\frac{f-h}2}a) && \mbox{occurs $1$ time,} \\
1+\frac1N(-\Delta_1+A\Delta_5+B\Delta_6)&&
\mbox{occurs $\frac{(l-1)(q-1)}{2l}$ times,}\\
1+\frac1N(-\Delta_1+A\Delta_6+B\Delta_5)&&
\mbox{occurs $\frac{(l-1)(q-1)}{2l}$ times,}\\
1+\frac1N(-1+\frac{l-1}2(A+B))\Delta_1 &&
\mbox{occurs $\frac{q-1}{l^2}$ times,}\\ +\frac{l-1}{2l}p^{\frac{f-h}2}\Delta_2-\frac{l-1}{2l}(A+B+2) && \\
1+\frac1N(-1+\frac{l-1}2(A+B))\Delta_1&&
\mbox{occurs $\frac{(l-1)(q-1)}{2l^2}$ times,}\\+\frac{1}{l}p^{\frac{f-h}2}\Delta_3-\frac{l-1}{2l}(A+B+2)&&\\
1+\frac1N(-1+\frac{l-1}2(A+B))\Delta_1&&
\mbox{occurs $\frac{(l-1)(q-1)}{2l^2}$ times,}\\+\frac{1}{l}p^{\frac{f-h}2}\Delta_4-\frac{l-1}{2l}(A+B+2) &&
\end{eqnarray*}
where $$\Delta_1=\zeta_p^{\frac{l(l-1)}2}+\frac{l-1}2\zeta_p^{l\varepsilon}+\frac{l-1}2\zeta_p^{-l(1+\varepsilon)},$$
$$\Delta_2=a\zeta_p^{\frac{l(l-1)}2}+\frac{-a+bl}2\zeta_p^{l\varepsilon}+\frac{-a-bl}2\zeta_p^{-l(1+\varepsilon)},$$
$$\Delta_3=\frac{-a+bl}2\zeta_p^{\frac{l(l-1)}2}+\frac{a-al-2bl}4\zeta_p^{l\varepsilon}+\frac{a(l+1)}4\zeta_p^{-l(1+\varepsilon)},$$
$$\Delta_4=\frac{-a-bl}2\zeta_p^{\frac{l(l-1)}2}+\frac{a(l+1)}4\zeta_p^{l\varepsilon}+\frac{a-al+2bl}4\zeta_p^{-l(1+\varepsilon)},$$
$$\Delta_5\frac{-1+\sqrt{-l}}2\zeta_p^{\frac{l(l-1)}2}+\frac{(l-1)(-1+\sqrt{-l})}4(\zeta_p^{l\varepsilon}+\zeta_p^{-l(1+\varepsilon)})+\frac{l(1-\sqrt{-l})}2,$$
$$\Delta_6=\frac{-1-\sqrt{-l}}2\zeta_p^{\frac{l(l-1)}2}+\frac{(l-1)(-1-\sqrt{-l})}4(\zeta_p^{l\varepsilon}+\zeta_p^{-l(1+\varepsilon)})+\frac{l(1+\sqrt{-l})}2,$$
$A=p^{\frac{f-hl}2}(\frac{a+b\sqrt{-l}}2)^l$, $B=p^{\frac{f-hl}2}(\frac{a-b\sqrt{-l}}2)^l$, $h$ is the ideal class number of $\mathbb{Q}(\sqrt{-l})$, and $a,b$ are integers given by (2.3).

\end{thm}

By (3.2) and Lemmas  3.5 and 3.6, we obtain  Theorem 3.8.
\begin{thm}Suppose that $ l\equiv 3\pmod 4$, $l\ne 3$,  be a prime. Let  $N=l^2$ be an integer and $f=\frac{l(l-1)}2$ be   the smallest positive integer such that $p^f\equiv 1 \pmod N$. Let $q=p^{f}$  for $\gcd(p,f)=1$.
Suppose that  $-l{\equiv} 1\pmod p$. Then the Walsh spectrum of    $f(x)=\Tr_{q/p}(x^{\frac{q-1}{N}})$ is given as follows:
\begin{eqnarray*}
Value&& Frequency\\
\frac{l-1}2(A+B+lp^{\frac{f-h}2}a) && \mbox{occurs $1$ time,} \\
\frac{l+1}{2N}(1-\zeta_p)(1+\frac{1+\sqrt{-l}}2A+\frac{1-\sqrt{-l}}2B)&&
\mbox{occurs $\frac{(l-1)(q-1)}{2l}$ times,}\\
\frac{l+1}{2N}(1-\zeta_p)(1+\frac{1-\sqrt{-l}}2A+\frac{1+\sqrt{-l}}2B)&&
\mbox{occurs $\frac{(l-1)(q-1)}{2l}$ times,}\\
\frac1{2N}(1-\zeta_p)(l+1-\frac{l^2-1}2 (A+B)
-\frac{(l-1)(a+\delta bl)}2lp^{\frac{f-h}2})&&
\mbox{occurs $\frac{q-1}{l^2}$ times,} \\
\frac1{2N}(1-\zeta_p)(l+1-\frac{l^2-1}2 (A+B)
+\frac{a-al+2\delta bl}2lp^{\frac{f-h}2})&&
\mbox{occurs $\frac{(l-1)(q-1)}{2l^2}$ times,}\\
\frac1{2N}(1-\zeta_p)(l+1-\frac{l^2-1}2 (A+B)
+\frac{a(l+1)}2lp^{\frac{f-h}2})&&
\mbox{occurs $\frac{(l-1)(q-1)}{2l^2}$ times,}
\end{eqnarray*}
where   $A=p^{\frac{f-hl}2}(\frac{a+b\sqrt{-l}}2)^l$, $B=p^{\frac{f-hl}2}(\frac{a-b\sqrt{-l}}2)^l$, $h$ is the ideal class number of $\mathbb{Q}(\sqrt{-l})$,  $a,b$ are integers given by (2.3), and $$\delta=\left\{\begin{array}{lll}
-1, &\mbox{ if } \sqrt {-l}\equiv 1\pmod{\mathcal P_1},\\
1,&\mbox{ if } \sqrt {-l}\equiv -1\pmod{\mathcal P_1}.\\
 \end{array}\right.$$
\end{thm}

Let $l$ is a prime congruent to $3$ modulo $4$, $l\neq 3$.  Suppose that  $1+l=4p^h$, where $h$ is the ideal class number of $\mathbb{Q}(\sqrt{-l})$.  Then in (2.3), $a,b$ take  $\pm  1$. By $a^2 + lb^2 = 4p^h $, $[(\mathbb{Z}/l\mathbb{Z})^\ast: \langle p\rangle] = 2$, and $a\equiv -2p^{\frac{l-1+2h}4} \pmod l$,   $$a=\left\{\begin{array}{lll}
1, &\mbox{ if } l\equiv 3\pmod 8,\\
-1,&\mbox{ if } l\equiv 7\pmod 8.\\
 \end{array}\right.$$
 While $b$ can be determined up to sign.

In  Theorem 3.8, suppose that $\sqrt {-l}\equiv 1\pmod{\mathcal P_1}$ for some $p$. If  $l\equiv 3\pmod 8$, then take $b=-1$; if  $l\equiv 7\pmod 8$, then take $b=1$.  Suppose that $\sqrt {-l}\equiv -1\pmod{\mathcal P_1}$ for some $p$. If  $l\equiv 3\pmod 8$, then  take $b=1$; if  $l\equiv 7\pmod 8$, then take $b=-1$.

Without loss of generality,  we  give  Corollary 3.9 based on  one of above conditions.
\begin{cor} The notations are as Theorem 3.8. Suppose that $1+l=4p^h$ and $\sqrt {-l}\equiv 1\pmod{\mathcal P_1}$ for some prime $p$. If  $l\equiv 3\pmod 8$, then the Walsh spectrum of   function $f(x)=\Tr_{q/p}(x^{\frac{q-1}{N}})$ is given as follows:
\begin{eqnarray*}
Value&& Frequency\\
\frac{l-1}2(A+B+lp^{\frac{f-h}2}) && \mbox{occurs $1$ time,} \\
\frac{l+1}{2N}(1-\zeta_p)(1+\frac{1+\sqrt{-l}}2A+\frac{1-\sqrt{-l}}2B)&&
\mbox{occurs $\frac{(l-1)(q-1)}{2l}$ times,}\\
\frac{l+1}{2N}(1-\zeta_p)(1+\frac{1-\sqrt{-l}}2A+\frac{1+\sqrt{-l}}2B)&&
\mbox{occurs $\frac{(l-1)(q-1)}{2l}$ times,}\\
\frac{l+1}{4N}(1-\zeta_p)(2-(l-1) (A+B+lp^{\frac{f-h}2}))&&
\mbox{occurs $\frac{q-1}{l^2}$ times,} \\
\frac{l+1}{4N}(1-\zeta_p)(2-(l-1) (A+B)+ lp^{\frac{f-h}2})&&
\mbox{occurs $\frac{(l-1)(q-1)}{l^2}$ times,}
\end{eqnarray*}
where   $A=p^{\frac{f-hl}2}(\frac{1-\sqrt{-l}}2)^l$,  $B=p^{\frac{f-hl}2}(\frac{1+\sqrt{-l}}2)^l$, and $h$ is the ideal class number of $\mathbb{Q}(\sqrt{-l})$.
\end{cor}

\begin{exa} Let $p=2$ and $l=7$.   It is straightforward to check that
$\ord_{7^2} (2) = 21 = \frac{\phi(7^2)}2=f$. The class number $h$ of $\Bbb Q(\sqrt{-7})$ is  1 (see \cite[ P.514]{C}).  Therefore we indeed have $1+l=4p^h$
 in
this case. From Corollary 3.19,  the Walsh spectrum of $f(x)=\Tr_{2^{21}/2}(x^{\frac{2^{21}-1}{7}})$  over $\Bbb F_{2^{21}}$ is five-valued.
\end{exa}

\begin{exa} Let $p=3$ and $l=107$.  It is straightforward to check that
$\ord_{ 107^2} (3) = 5671 = \frac{\phi(107^2)}2=f$.  Let $q=3^{5671}$. The class number $ h$ of $\Bbb Q(\sqrt{-107})$
 is  3 (see \cite[P.514]{C}).  Therefore we also have $1+l=4p^h$
 in this case. From Corollary 3.9,  the Walsh spectrum of $f(x)=\Tr_{q/3}(x^{\frac{q-1}{107^2}})$ over $\Bbb F_q$ is five-valued.

\end{exa}

\section{Concluding remarks}

   In this paper, we  investigated the Walsh spectrum of $f(x)=\Tr_{q/p}(x^{\frac{q-1}{N}})$  over $\Bbb F_q$ in index two case, where $l\equiv 3\pmod 4$, $l\ne 3$, is a prime, $N=l^2$,   the  order of a prime $p$ modulo $N$ is $f=\frac{l(l-1)}2$,   and $q=p^f$. In particular, a class of monomial functions with five-valued Walsh spectrum are presented  if $l$ has the form $l+1=4p^h$, where $h$  is the ideal class number of $\mathbb{Q}(\sqrt{-l})$.

    In fact, the Walsh spectrum of $f(x)=\Tr_{q/p}(x^{\frac{q-1}{N}})$ for the general case $N=l^m$ and $m> 2$ can also be  settled  by the same  method presented in this paper. But the process is more complex and the Walsh spectrum may have many values.
\subsection*{Acknowledgments}
The authors are very grateful to the reviewers and the Editor for their valuable  suggestions that much improved the quality of this paper.

\end{document}